\documentclass[12pt,tightenlines,eqsecnum,floats,showpacs,nofootinbib,amsmath,amssymb,aps,prd]{revtex4-2}

 \usepackage{tikz}
\usepackage{graphicx,verbatim}
\usepackage{amsmath}
\usepackage{amsfonts}
\usepackage{amssymb}
\usepackage{amsthm}

\usepackage{psfrag}
\usepackage{accents}

\usepackage{color}

\DeclareMathOperator{\soo}{so}

\usepackage[colorinlistoftodos]{todonotes}
\usepackage{epstopdf}
\usepackage{soul}

\def\Lie{\mathcal{L}}

\newtheorem{theorem}{Theorem}[section]
\newtheorem{lemma}[theorem]{Lemma}
\newtheorem{prop}[theorem]{Proposition}
\theoremstyle{definition}
\newtheorem{df}{Definition}[section]

\newcommand{\pb}[1]{\hbox{\lower0.5ex\hbox{${}_{\leftarrow}$}}\kern-1.9ex{#1}}

\def\={\hat{=}}

\def\be{\begin{equation}}
\def\ee{\end{equation}}
\def\ba{\begin{eqnarray}}
\def\ea{\end{eqnarray}}

\def\SO(3){\rm SO(3)}
\def\so(3){\rm so(3)}
\def\SO(4){\rm SO(4)}
\def\so(4){\rm so(4)}
\def\SO(1,4){\rm SO(1,4)}
\def\so(1,4){\rm so(1,4)}
\def\SU(2){\rm SU(2)}

\begin{document}
\title{Symmetries of the asymptotically de Sitter spacetimes}
\author{Wojciech Kamiński, Maciej Kolanowski and Jerzy Lewandowski}
\affiliation{Institute of Theoretical Physics, Faculty of Physics, University of Warsaw, Pasteura 5, 02-093 Warsaw, Poland}
\begin{abstract} \noindent
    We start a systematic investigation of possible isometries of the asymptotically de Sitter solutions to Einstein equations. We reformulate the Killing equation as conformal equations for the initial data at $\mathcal{I}^+$. This allows for partial classification of possible symmetry algebras. In particular, if they are not maximal, they may be at most $4$-dimensional. We provide several examples. As a simple collorary it is shown that the only spacetime in which the Killing horizon intersects $\mathcal{I}^+$ (after a conformal completion) is locally the de Sitter universe.
\end{abstract}
\maketitle
\section{Introduction} \noindent
It is quite well-established nowadays that we live in the universe with a positive cosmological constant \cite{riess1998observational}. Despite that, we are still lacking the universal framework for the description of the gravitational radiation and different astrophysical phenomena when $\Lambda > 0$. Not surprisingly, the last few years witnessed a growing interest in the topic -- many different approaches, definitions and solutions were presented (see for example \cite{kastor2002positive, Penrose:2011zza, Ashtekar:2015lla, Szabados:2015wqa, Chrusciel:2020rlz, Ashtekar:2014zfa, He:2015wfa, Poole:2018koa, Compere:2019bua, Ashtekar:2015lxa, Bishop:2015kay, Date:2015kma, Kolanowski:2020wfg, Kolanowski:2021hwo}. The fundamental reason for the difficulty of the task lies within the nature of the null infinity $\mathcal{I}^+$. When $\Lambda>0$, the null infinity is spacelike and thus there is no notion of the asymptotic time translation and subsequently of the positive-definite hamiltonian. Moreover, $\mathcal{I}^+$ carries no natural structure besides that of a smooth manifold which leads to the conclusion that all diffeomorphisms of $\mathcal{I}^+$ are asymptotic symmetries which renders the notion rather useless. We do not plan to offer solution to those hard and important issues. The scope of this note is far more modest. We aim to begin a systematic investigations of the isometries of the asymptotically de Sitter spacetimes. Of course, metrics with many symmetries are rather rare and cannot properly describe physical processes like e.g. merger of two black holes. Nevertheless, they may be quite useful, for example as a starting point for the perturbation theory.\\
In this work we will be concerned with the vacuum asymptotically de Sitter spacetimes. We will show that they may admit (locally) $10$-dimensional algebras of symmetries (in which case it is locally de Sitter) or the algebra is $4, 3,2,1$-dimensional in which case it can be chosen to act by isometries on the boundary data. Moreover, we will classify all maximally symmetric and almost maximally symmetric solutions with $\mathcal{I}^+ = S^3$ -- they corresponds to the de Sitter and Taub-NUT-de Sitter spacetimes, respectively. Although we do not have a full classification of all possible 3d algebras, such classification seem to be feasible (although rather long - in particular, all Bianchi algebras would appear as special cases). \\
Most approaches to the asymptotically de Sitter spacetimes rely heavily {\it nomen omen} on the asymptotic behavior of the fields and on the existence of the null infinity $\mathcal{I}^+$. We are going to use this structure to our benefit as well. However, let us mention a different paradigm. In \cite{Ashtekar:2019khv} it was proposed that one could use (certain generalization of) the cosmological horizon as a local null infinity. Further progress in this direction was presented in \cite{Ashtekar:2021kqj, Ashtekar:2021wld} where the BMS-like symmetry group of such horizons were identified and used to calculate physical quantities (multipole moments, charges and their fluxed) within the perturbation theory. In the context of our current investigations it is rather natural to ask about the Killing horizons associated with the symmetries described in the previous paragraph. One striking property of the de Sitter spacetime is the fact that each Killing horizon can be identified with a lightcone emanating from a point on $\mathcal{I}^+$ in a conformally completed spacetime. The question that arises is how common is the presence of a Killing horizon intersecting $\mathcal{I}^+$ in the asymptotically de Sitter spacetimes. This is the next issue we address in this paper and find its precise solution -- this property is unique to the de Sitter spacetime. Of course, it does not mean that different solutions do not have Killing horizons (the simplest example being the Schwarzschild-de Sitter black hole) but rather, using a picturesque language, that 'point of intersection' lies at the 'infinity of $\mathcal{I}^+$. This is heavily tied with the topological properties of $\mathcal{I}^+$ for different solutions, as analyzed in many examples in \cite{Ashtekar:2014zfa}. \\
Incidentally, both questions can be analyzed in the same language - by the extensive usage of the conformal properties of the Cauchy data. In fact, we will start with the latter problem simply because it is easier and it naturally sets up the stage for the classification.
\section{Killing vector fields at $\mathcal{I}^+$} \noindent
A prerequisite for a Killing horizon to exist is the existence of a Killing vector field (KVF) itself. Thus, we will start by showing how they can be read off from the initial value formulation in the spirit of Friedrich \cite{friedrich1986}. Asymptotically de Sitter spacetimes can be put into the Fefferman-Graham gauge:
\begin{equation}\label{gFG}
    g = -\frac{\ell^2 d\rho^2}{\rho^2} + \frac{\ell^2 q_{ab} dx^a dx^b}{\rho^2},
\end{equation}
where $q_{ab} = q_{ab}(\rho, x^c)$ is smooth and $\ell = \sqrt{\frac{3}{\Lambda}}$. The null infinity $\mathcal{I}^+$ corresponds to $\rho = 0$. Let us denote
\begin{equation}
    q_{ab}^{(n)} = \frac{1}{n!}\frac{\partial^n}{\partial\rho^n} q_{ab}|_{\rho=0}.
\end{equation}
From the vacuum Einstein equations, 
\begin{equation}\label{EL}
  R_{\alpha\beta}=\Lambda g_{\alpha\beta}, \ \ \ \ \ \Lambda = \frac{3}{\ell^2} 
\end{equation}
the following constraints follow
\begin{align}
    q^{(1)}_{ab} &= 0 \label{q1} \\
    q^{(2)}_{ab} &= \mathring{R}_{ab} - \frac{1}{4}\mathring{R} q_{ab}^{(0)} \label{q2} \\
    q^{(0)ab}q^{(3)}_{ab} &= 0 \label{cons1} \\
    D^a q^{(3)}_{ab} &=0, \label{cons2}
\end{align}
where $\mathring{R}_{ab}$ and $\mathring{R}$ are Ricci tensor and scalar (respectively) of $q^{(0)}$ and $D$ is its covariant derivative. All $q^{(n)}_{n>3}$ are recursively given as functions of $q^{(0)}$ and $q^{(3)}$. One can easily notice that $q^{(2)}$ is twice the Schouten tensor of $q^{(0)}$. Notice that since $q_{ab}$ is merely smooth, it is only an asymptotic  expansion and not necessary a convergent series.  Nevertheless, it follows from \cite{friedrich1986} that the proper initial data at $\mathcal{I}^+$ are classes $[(q^{(0)}, q^{(3)})]$, where two pairs are equivalent if they are connected by a (non-singular) conformal transformation:
\begin{equation}\label{data}
    (q^{(0)}, q^{(3)}) \sim (\omega^2 q^{(0)}, \omega^{-1} q^{(3)}). %\label{conf_equiv}
\end{equation}
One can easily checked that the constraints for $q^{(3)}$ are conformally covariant.
To agree with a standard notation, we will introduce a holographic energy-momentum tensor
\begin{equation}
    T_{ab} = \frac{\sqrt{3\Lambda}}{16\pi G}q^{(3)}_{ab}.
\end{equation}
Since we are mainly interested in the behavior \emph{at} $\rho = 0$, it is convenient to introduce an unphysical metric
\begin{equation}
    \hat{g} = \ell^{-2}\rho^2 g,
\end{equation}
which is smooth at all values of $\rho$. \\
We are looking for the KVFs $X$:
\begin{equation}
    \mathcal{L}_X g_{ab} = 0.
\end{equation}
This equation can be rewritten in a more convenient way as
\begin{equation}
    \mathcal{L}_X \hat{g}_{ab} = \frac{2\mathcal{L}_X \rho}{\rho} \hat{g}_{ab}. \label{Killing_eq}
\end{equation}
Let us evaluate first its $\rho \rho$ component:
\begin{equation}
    X^\rho_{\rho} = \rho^{-1} X^\rho.
\end{equation}
The solution is immediate:
\begin{equation}
    X^\rho = \rho \mathring{X}^\rho,
\end{equation}
where $\mathring{X}^\rho_{,\rho} = 0$. Thus, it follows that a KVF must be tangent to $\mathcal{I}^+$. Let $\xi \in \Gamma(T\mathcal{I}^+)$ denotes restriction of $X$ to $\mathcal{I}^+$. We will now solve the remaining equations perturbatively in $\rho$. Evaluating $ab$ components of \eqref{Killing_eq} at $\rho =0$ gives
\begin{equation}
    \mathcal{L}_\xi q^{(0)}_{ab} = 2 \mathring{X}^\rho q^{(0)}_{ab}.
\end{equation}
Thus, we see that $\xi$ is a conformal Killing vector field (CKVF) of $q^{(0)}$ and 
\begin{equation}
    \mathring{X}^\rho = \frac{1}{3} D_a \xi^a.
\end{equation}
$\rho a$ components of \eqref{Killing_eq} gives
\begin{equation}
    X^b_\rho q_{ab} - X^\rho_{,a} = 0,
\end{equation}
which can be immediately solved
\begin{equation}
    X^b = \xi^b + \int_0^\rho \rho' q^{ab}(\rho', x^c) d\rho' \left( D_m \xi^m \right)_{,a}.
\end{equation}
Thus, every KVF is uniquely defined by its restriction to $\mathcal{I}^+$:
\begin{equation}
    X = \frac{1}{3}\rho  D_m \xi^m \partial_\rho + \left( \xi^b + \int_0^\rho \rho' q^{ab}(\rho', x^c) d\rho' \left( D_m \xi^m \right)_{,a} \right) \partial_b. \label{killing}
\end{equation}
Of course, even when $\xi$ is a CKVF of $q^{(0)}$, \eqref{killing} is not necessarily going to be a KVF of $g$. To make sure that it is so, we need to evaluate \eqref{Killing_eq} in higher orders of $\rho$ -- to this end we will differentiate it with respect to $\rho$. The first derivative of both sides automatically vanish since $q^{(1)} = 0$. The second one reads:
\begin{equation}
    \mathcal{L}_\xi q^{(2)}_{ab} = - \frac{1}{3} D_{(a}D_{b)} D_m \xi^m
\end{equation}
which is just a geometrical identity whenever $q^{(2)}$ satisfies \eqref{q2}. Notice that $X^a_{,\rho}|_{\rho =0} = X^a_{,\rho \rho \rho}|_{\rho =0} = 0$ due to the form \eqref{killing} and \eqref{q1}. Thanks to that, the third derivative of \eqref{Killing_eq} is extremely simple, and it reads
\begin{equation}
    \mathcal{L}_\xi T_{ab} = - \frac{1}{3} D_m \xi^m T_{ab}. 
\end{equation}
This should not be a surprise since $(q_{ab}, T_{ab})$ are defined up to a common conformal transformation. One can show that if $\xi$ satisfies conditions derived above, then an associated $X$ given by \eqref{killing} is truly the KVF of $g$. \\
Indeed, let us consider one parameter family of the local initial data
\begin{equation}
     (\omega_s^2 \phi_s^*q^{(0)}, \omega_s^{-1} \phi_s^*q^{(3)}),\quad \omega_0=1,\ \phi_0={\rm id}
\end{equation}
for some conformal factors $\omega_s$ and local diffeomorphisms $\phi_s$. Uniqueness results of \cite{friedrich1986} show that there exists one parameter family of diffeomorphisms $\Phi_s$ of the spacetime which transform the development of the data at $s$ into the one for $s=0$. Suppose that $X=\frac{d\phi_s}{ds}$ is CKVFs symmetry of the initial data. We can locally integrate it and choose a conformal factor such that $(\omega_s^2 \phi_s^*q^{(0)}, \omega_s^{-1} \phi_s^*q^{(3)})$ is independent of $s$. In this situation, the derivative $\frac{d\Phi_s}{ds}$ defines the corresponding  KVF of the spacetime. \\
We have thus reformulated the problem of finding isometries to the problem of identifying conformal symmetries of the initial data:
\begin{align}
    \begin{split}
        \mathcal{L}_\xi q_{ab} &= 2 \omega q_{ab} \\
        \mathcal{L}_\xi T_{ab} &= - \omega T_{ab}.
    \end{split}
\end{align}
We will use those equations extensively in the following sections.
\section{Killing horizons} \noindent
Thus far, we have only focused on the existence of KVFs for themselves. However, our main object of interest is a \emph{Killing horizon} $H$, it means such a null hypersurface that $X$ is its null normal. Moreover, we assume that after a conformal completion, an intersection of $H$ and $\mathcal{I}^+$ is non-empty, just as it has place in the de Sitter spacetime. Since $X$ is tangent to $\mathcal{I}^+$, it is spacelike and its length (say, in $\hat{g}$ metric) is nonnegative. On the other hand, its length on $H$ vanishes. Thus, at the (non-empty) intersection, $X$ must vanish. \\
All Killing horizons are nonexpanding, it means all cross-sections of $H$ have the same area\footnote{We assume that $H$ has a product topology $\mathbb{R} \times K$, where $K$ is a compact surface, and thus it makes sense to talk about a cross-section's area. An example of a situation in which it does not hold and one has well-defined $\mathcal{I}^+$ can be given by the Kerr-Taub-NUT-de Sitter. Then, the horizon is topologically $S^3$. Nevertheless, even in this example the horizon is separated from from $\mathcal{I}^+$ \cite{Lewandowski:2021qty}.}. In particular, we can take $\rho = \textrm{const.}$ cross-sections and their area (in $g$ metric) is $\rho$-independent. On the other hand, their area in $\hat{g}$ decreases as $\rho^2$. In particular, it vanishes on $\mathcal{I}^+$. Thus, $H \cap \mathcal{I}^+$ cannot be a $2$-surface, so it is either a curve or a point\footnote{Our considerations are local on $\mathcal{I}^+$, and so we can focus only on one connected component of the intersection. From that point, we will assume without loss of generality $H \cap \mathcal{I}^+$ to be connected}. As we will show, in a moment, the first case is excluded. \\ 
Let us first assume that $\xi$ is a KVF of $q^{(0)}$. Then, it follows from \eqref{killing} that
\begin{equation}
    X = \xi^a \partial_a,
\end{equation}
and so it is spacelike everywhere (or vanishing) in the domain in which coordinates $(\rho, x^c)$ are well-defined. In particular, it cannot have a horizon. As a next step, let us assume that it is not a KVF, but there exists a positive function $\omega$ on a neighborhood of $H \cap \mathcal{I}^+$ such that 
\begin{equation}\label{essential-xi}
    \mathcal{L}_\xi \omega^2 q^{(0)} = 0.
\end{equation}
Using terminology from \cite{frances2012local} we introduce a definition:

\begin{df}
Conformal vector field $\xi$ is called non-essential if there exists nonzero $\omega$ satisfying \eqref{essential-xi}. Otherwise, it is called essential.
\end{df}
We will also use local version of these definitions. Vector field is essential at $x$ if it is essential for every neighborhood of $x$.

\noindent
Such $\xi$ are called non-essentials CKVF. If $\xi$ is non-essential, we can consider a spacetime $(M', g')$ with an initial data $(\omega^2 q^{(0)}_{ab}, \omega^{-1} T_{ab})$ which is diffeomorphic to the one we considered so far. It follows that $\xi$ is spacelike on $(M', g')$ and so also on $(M,g)$. Thus, if $\xi$ is a non-essential CKVF of $q^{(0)}$, it does not have a Killing horizon intersecting with $\mathcal{I}^+$ and we are left only with an essential case. 
Fortunately, if a Riemannian manifold possess an essential CKVF for the point $x_0$, it is locally conformally flat  around this point \cite{frances2012local}.
Without loss of generality, we may assume that $q^{(0)}_{ab} = \delta_{ab}$. One can easily check that if $\xi$ is an essential CKVF on every neighborhood of a given point, then this point is an isolated zero of $\xi$. Thus, as promised, it follows that (each connected component of) $H \cap \mathcal{I}^+$ is a point.
\subsection{Holographic energy--momentum tensor} \noindent
Thus far, we were able to establish that $q^{(0)}$ is conformally flat. We assumed without loss of generality that it is simply flat around $H \cap \mathcal{I}^+$. The only remaining part of the initial data is thus $T_{ab}$. \\
It is easier to deal with scalars rather than tensors, so let us introduce
\begin{equation}\label{thescalar}
    \chi = q^{(0) ac} q^{(0) bd} T_{ab} T_{cd} \left( q^{(0)}_{ij} \xi^i \xi^j \right)^3.
\end{equation}
It is easy to notice that $\chi$ is a smooth function satisfying
\begin{equation}
    \mathcal{L}_\xi \chi = 0 \label{const_funct}.
\end{equation}
Thus, it is constant along the integral lines of $\xi$. As we already emphasized, we are interested in local symmetries and so let us assume that \eqref{const_funct} holds in a ball of radius $\epsilon > 0$ around the fixed point of $\xi$. The following useful lemma holds:
\begin{lemma} \label{lemat}
Let $x_0$ be a fixed point of the CKVF $\xi$.
For any open neighborhood $x_0\in U$, there exists an open subset $x_0\in O\subset U$ such that for any $x\in O$
either the integral line of $\xi$ starting at $x$ or the integral line ending at this point is fully contained in $O$. Moreover, if $\xi$ is an essential CKVF, the same starting or ending line approaches the fixed point asymptotically.
\end{lemma} 
\begin{proof}
If the metric is not conformally flat, then the point $x_0$ is not essential and there exists a metric $q^{(1)}$ which is preserved by $\xi$. The ball $B(x_0,\epsilon)$ is preserved by the flow of the Killing vector field. We will now focus on the conformally flat case. \\
Let us introduce cartesian coordinates centered around $x_0$. Since we know the form of $q^{(0)}$, we know its CKVFs as well:
\begin{equation}
    \xi = \left( p^i + r^i_{\ j} x^j +  S x^i  + 2K_j x^j x^i - x_j x^j K^i \right) \partial_{x^i},
\end{equation}
where $p^i, r^i_{\ j}, S, K^i$ are covariantly constant and $r_{ij}$ is antisymmetric. Since $0$ is a zero of $\xi$, we put $p^i = 0$. Moreover, we can always put $\xi$ into one of the following three standard forms:
\begin{enumerate}
    \item $\xi=\left( r^i_{\ j} x^j +  2K_j x^j x^i - x_j x^j K^i \right) \partial_{x^i}$ where $r^i_{\ j}K_i=0$ and the vector $K\not=0$
    \item $\xi = \left(r^i_{\ j} x^j +  S x^i \right) \partial_{x^i}$, where $S\not=0$
    \item $\xi = r^i_{\ j} x^j \partial_{x^i}$.
\end{enumerate}
We can restrict attention to vectors in these standard forms. If $\xi=r^i_{\ j} x^j \partial_{x^i}$ then it is not essential since it preserves $q^{(0)}$. Its integral line are circle around $x_0$ and thus are surely contained in a small ball. For the two remaining vector fields, we can show that the thesis of the lemma holds for $|x|\leq \epsilon$ for any $\epsilon>0$. In order to show this,
we apply inversion. We then need to show that either the forward or the backward flow for a given $\xi'$ preserves $|x|>\epsilon^{-1}$ and converges to infinity.
\begin{enumerate}
    \item For the vector field $\xi'=\left( r^i_{\ j} x^j +  S x^i \right) \partial_{x^i}$, we have
    $\Lie_\xi |x|=S$ thus $|x|\rightarrow \infty$ for either the forward flow ($S>0$) or the backward flow ($S<0$). Moreover the space $|x|>\epsilon^{-1}$ is preserved.
    \item In case of the vector field
    $\xi'=\left( K^i + r^i_{\ j} x^j  \right) \partial_{x^i}$ then 
    \begin{equation}
        \Lie_\xi K\cdot x=|K|^2>0,\quad \Lie_\xi |x|^2=2K\cdot x.
    \end{equation}
    The space $K\cdot x\geq 0$ and $|x|>\epsilon^{-1}$ is preserved by the forward flow and moreover $|x|\rightarrow \infty$. Similarly, for the space $K\cdot x\leq 0$ and $|x|>\epsilon^{-1}$.
\end{enumerate}
\end{proof} \noindent
Since $\chi$ vanishes at the fixed point and is constant along every integral line, it vanishes in the whole ball and so does $T_{ab}$ (since that fixed point is an isolated zero of $\xi$). Thus, there is a neighborhood of the fixed point, on which the initial data are those of the de Sitter spacetime. Since the Cauchy problem is well-posed, its past development is isometric to the de Sitter spacetime. Let us formulate this result as a theorem:
\begin{theorem}\label{Theorem}
    Let $H$ be a Killing horizon which intersects the scri $\mathcal{I}^+$  defined by $\rho=0$ in a (conformal completion) of a spacetime endowed with a physical metric tensor (\ref{gFG}) that satisfies the vacuum Einstein equations (\ref{EL}). Then, there is a neighborhood of that intersection in which the physical metric is isometric to the de Sitter one.
\end{theorem}
Our proof that $T_{ab}$ vanishes can be easily generalized to show that there is no tensor $S_{ab}$ satisfying
    \begin{equation}
        \mathcal{L}_\xi S_{ab} = \frac{s}{3} D_m \xi^m S_{ab}
    \end{equation}
    when $s < 2$. To this end, we just need to consider the following function:
    \begin{equation}
        \chi = S^{ab} S_{ab} \left( \delta_{ij} \xi^i \xi^j \right)^{2 -s}
    \end{equation}
    which is (at least) continuous at the origin, smooth everywhere else, and is constant along the integral lines of $\xi$.
\section{Local symmetries}

\noindent We will now extend the characterization of the essential conformal structures to the case with tensor $T_{ab}$. Let ${\mathcal V}$ be the algebra of CKVF symmetries of the neighborhood of $x\in{\mathcal I}^+$, that is $\forall \xi\in{\mathcal V}$
\begin{align}
    \mathcal{L}_\xi q^{(0)}_{ab} &= 2 \alpha_\xi q^{(0)}_{ab} \\
    \mathcal{L}_\xi T_{ab} &= - \alpha_\xi T_{ab}, \label{cond_T}
\end{align}
We introduce
\begin{df}
Algebra of symmetries ${\mathcal V}$ is non-essential at $x$ if there exists an open neighborhood of $x$ and a conformal factor $\omega^2$ such that
\begin{equation}
    \mathcal{L}_\xi \omega^2q^{(0)}_{ab} =0,\quad
    \mathcal{L}_\xi \omega^{-1}T_{ab} =0
\end{equation}
Otherwise, we call it essential.
\end{df} \noindent
Our goal is to extend a result of \cite{frances2012local}. Namely,

\begin{prop}\label{lm:V-X}
Suppose ${\mathcal V}$ is essential at $x$. Then, there exists $X\in {\mathcal V}$ which is essential at $x$.
\end{prop} \noindent
Let us remark, that it is not obvious that essentiality could in principle follows from properties of the whole algebra. Our proof of Proposition \ref{lm:V-X} will be based on the following result:
\begin{lemma}\label{lm:proper}
Let ${\mathcal V}$ be an subalgebra of conformal vector fields. Denote by ${\mathcal V}_x$ a stabilizer of $x$ (a subalgebra of vectors  vanishing at $x$). If ${\mathcal V}_x$ is non-essential then it is also true for ${\mathcal V}$.
\end{lemma}

\noindent{Remark}: It is a local version of property of the proper action (see \cite{alekseevskiui1972groups}).

\begin{proof}
Let $\tilde{G}$ be a universal group generated by the algebra ${\mathcal V}$ with a subgroup $\tilde{K}$ generated by ${\mathcal V}_x$. There exists a representative $g^0$ in a conformal class, which is preserved by ${\mathcal V}_x$. 
We can consider a ball in this metric $B(x,\eta_1)$ for small enough $\eta_1$ such that it is inside normal coordinates chart around $x$. As $x$ is a fixed point, it is preserved by ${\mathcal V}_x$. We can integrate the action of the algebra ${\mathcal V}_x$ to the action of a group $\tilde{K}$. Let $H$ be a subgroup of $\tilde{K}$ which acts trivially on this ball. Let us notice that $H$ also acts trivially on ${\mathcal V}$ and thus it is normal as it belongs to the center of $\tilde{G}$. We may consider
\begin{equation}
    G=\tilde{G}/H,\quad K=\tilde{K}/H.
\end{equation}
We notice that $K$ is a compact group (by an injective homomorphism $K\rightarrow {\rm End}(T_xM)$ we can identify it with a subgroup of $SO(3)$).
We will consider local action of $G$. Let $L$ be a complementary subspace to ${\mathcal V}_x$ in ${\mathcal V}$. We equip it with an auxiliary norm. 
\begin{enumerate}
    \item There exists $\epsilon_1>0$ such that
    \begin{equation}
    \{l\in L\colon |l|<\epsilon_1\}\times K\ni(l,k)\rightarrow \exp l\cdot k\in G    
    \end{equation}
     is injective. This is thanks to the compactness of $K$. We denote the pull-back of the left Haar measure by $d\mu$.
    \item There exists $0<\epsilon_2\leq \epsilon_1$ such that $(\exp l)x\in B(x,\eta_1)$ for all $l\in L$ such that $|l|<\epsilon_2$.
    \item As no nonzero vector in $L$ vanishes in $x$, there exists $\epsilon_3<\frac{1}{2}\epsilon_2$ and $\eta_2<\eta_1$ such that 
    $d((\exp l)x,x)>\eta_2$ for $l\in L$ satisfying $\epsilon_3/2<|l|<2\epsilon_3$.
\end{enumerate}
From continuity of the action there exists $\epsilon_4<\epsilon_3$ and $\eta_3<\eta_2$ such that for all $k\in K$ and $l\in L$ satisfying $\epsilon_4/2<|l|<2\epsilon_4$ it holds
\begin{equation}
\exp l\cdot k(B(x,\eta_3))\cap B(x,\eta_3)=\emptyset
\end{equation}
Let us define 
\begin{equation}
    g^1_{ab}(y)=f\left(\frac{d(x,y)}{\eta_3}\right)g^0_{ab}(y)
\end{equation}
where $f$ is a smooth function which is one in zero and vanishes for all argument bigger equal than 1 (and nowhere else). This is a $K$ invariant metric in $B(x,\eta_3)$. We define a metric by integration
\begin{equation}
    g^{2}_{ab}=\int_{l\in L\colon |l|\leq \epsilon_4}\int_K d\mu\ \left(\exp l\cdot k\right)^*g^1_{ab}
\end{equation}
This metric is invariant under ${\mathcal V}$ in some small neighborhood of $x$. In fact the integral changes under the action of the algebra only by boundary terms (due to invariance of the Haar measure). However, the boundary does not contribute to the tensor in small neighbourhood of $x$. Thus ${\mathcal V}$ is non-essential.
\end{proof}

\begin{proof}[Proof of \ref{lm:V-X}]
From Lemma \ref{lm:proper} we know that ${\mathcal V}_x$ is essential. However, then
by Theorem 7.1 from \cite{frances2012local}, the neighborhood of $x$ is conformally flat. If every $X\in {\mathcal V}_X$ is non-essential then after applying inversion $X'$ is an element of Euclidean-Poincare transformation and it has a fixed point.  From Lemma 7.2 of \cite{frances2012local} we then know that there is a common fixed point for all vectors. Applying inversion with respect to this point we see that we transformed all vector fields in ${\mathcal V}$ into Euclidean-Poincare vector fields.
\end{proof}

The main result can be summarized as follows:

\begin{theorem}\label{thm:local}
Let ${\mathcal V}$ be an algebra of local conformal symmetries at $x$ of the data $(q_{ab},T_{ab})$. Then 
\begin{enumerate}
    \item \label{case:1-th-loc} either there exists a choice of non-vanishing $\omega$ in some neighborhood of $x$ such that
    \begin{equation}
        {\mathcal L}_\xi \omega^2q_{ab}=0,\ {\mathcal L}_\xi \omega^{-1}T_{ab}=0,\quad \forall \xi\in{\mathcal V}
    \end{equation}
    \item \label{case:2-th-loc} or the metric is conformally flat and $T_{ab}=0$ in the neighborhood of $x$ and at least one of the vectors in the algebra is essential at $x$.
\end{enumerate}
\end{theorem}

In the second case we will say that the data is locally de Sitter.

\begin{proof}
Either ${\mathcal V}$ is non-essential at $x$ and then the case \ref{case:1-th-loc} holds or there is an essential vector field in ${\mathcal V}$ and then by Lemma \ref{lemat} and the text below we are in the case \ref{case:2-th-loc}.
\end{proof}

\section{Global symmetries} \noindent
In this section, we will use our just gained knowledge of the initial data and their symmetries to classify what are possible isometries of the asymptotically de Sitter spacetime. To be more precise, we will show that the group of isometries of the asymptotically de Sitter spacetime can be only $0,1,2,3$ and $4$-dimensional, unless the spacetime is locally isomorphic to the de Sitter universe.\\
Let us start with some properties of the algebra of symmetries:
\begin{lemma}\label{lm:4-d}
Suppose that the algebra of conformal symmetries ${\mathcal V}$ is at least $4$  dimensional and non-essential at $x$. Then the dimension of ${\mathcal V}$ is either $4$ or $6$ and the algebra acts locally transitively around $x$. In the case of dimension $6$ the data is locally de Sitter around $x$.
\end{lemma}

\begin{proof}
As the algebra is non-essential we can assume it consists of Killing vector fields. We can identify the stabilizer ${\mathcal V}_x\subset \soo (3)$. As a Lie subalgebra, it can be either $3$ (full $\soo (3)$) or $1$ dimensional. \\
In the first case, it acts transitively on $T_xM$. However, it preserves the space 
\begin{equation}
Y=\{X(x)\colon X\in {\mathcal V}\},
\end{equation}
which is nontrivial as dimension of ${\mathcal V}$ is at least $4$. This means that $Y=T_xM$. Counting dimensions shows that $\dim {\mathcal V}=6$. However, this means that the metric is maximally symmetric and conformally flat. Additionally, $T_{ab}=cq_{ab}$ and as it is traceless $T_{ab}=0$. This is the case of locally de Sitter.

On the other hand, if ${\mathcal V}_x$ is one-dimensional then $Y=T_xM$, the dimension of ${\mathcal V}$ is $4$ and the algebra acts locally transitively.
\end{proof}

\begin{theorem}
If the spacetime is not everywhere locally de Sitter, then one of the following holds:
\begin{enumerate}
    \item The connected group of  symmetries is $4$-dimensional. It acts transitively on ${\mathcal I}^+$ and there is a choice of conformal class of the metric for which the action is by isometries.
    \item The group of symmetries is at most $3$-dimensional.
\end{enumerate}
\end{theorem}

\begin{proof}
We consider the connected component of the symmetry group and its Lie algebra ${\mathcal V}$.
Let us introduce the following scalar:
\begin{equation}
    \chi = T_{ab} T^{ab} + C_{abc} C^{abc},
\end{equation}
where $C_{abc}$ is a Cotton tensor of $q^{(0)}$. Since the Cotton tensor is a conformal invariant, it follows that $\chi$ has a conformal weight $-6$. If $\chi$ vanishes in a neighborhood of a point $x$ , then the data is locally de Sitter around this point (the metric is conformally flat and $T_{ab}$ vanishes). Suppose now that the point $x$ is non-essential and $\chi(x)=0$. Then by Lemma \ref{lm:4-d} the action of algebra is locally transitive and thus $\chi=0$ in the whole neighborhood of $x$. The same is true if $x$ is an essential point and so the set $\{\chi=0\}$ is open. As both open and closed it needs to be either an empty set or the whole ${\mathcal I}^+$.
\begin{enumerate}
\item If $\{\chi=0\}={\mathcal I}^+$ then the space is everywhere locally de Sitter.
\item If $\chi\not=0$ everywhere, then we can introduce equivalent data:
\begin{equation}
q'_{ab} = \chi^{\frac{1}{3}} q_{ab}, \quad
    T'_{ab} = \chi^{-\frac{1}{6}} T_{ab},    
\end{equation}
which satisfy for every $\xi\in{\mathcal V}$:
\begin{equation}
    \mathcal{L}_\xi q'_{ab} = 0, \quad
    \mathcal{L}_\xi T'_{ab} = 0
\end{equation}
and thus $\xi$ must be KVF of $q'$. From Lemma \ref{lm:4-d} the orbits of every point is the whole ${\mathcal I}^+$ (both orbit and its complement is an open set).
\end{enumerate}
Thus the metric $q_{ab}'$ is preserved by the connected component of the symmetry group.
\end{proof}

\section{Examples} \noindent
Thus far, we have shown that the possible dimensions of the isometry group are $d=4, 3, 2, 1$, assuming that the spacetime is not locally de Sitter. Of course, that does not prove that all of those cases are actually realized. What is left is to construct examples for each of those values. We will divide our discussion into different topologies of $\mathcal{I}^+$ which are commonly encountered. With the exception of $d=4$, we do not claim any sort of completeness.
\subsection{Sphere $S^3$} \noindent
Obviously, the first example that comes to one's mind is the global de Sitter spacetime. It is maximally symmetric and thus $d=10$. We already learned that when $d=4$, symmetries act transitively on $\mathcal{I}^+$ and (in an appropriate conformal frame) are isometries. Thus, $\mathcal{I}^+$ with a metric $q^{(0)}$ must be a homogeneous space. Fortunately, all homogeneous metrics on a simply connected 3-spaces are classified \cite{Patragenaru}. On $S^3$ they are simply given by squashed spheres:
\begin{equation}
    q^{(0)} = \lambda_1 \sigma_1^2 + \lambda_2 \sigma_2^2+\lambda_3 \sigma_3^2,
\end{equation}
where $\sigma_i$ are standard left invariant one-forms on $S^3$. When all $\lambda$ are different, this metric has $SU(2)$ symmetry ($d=3$). When two of them coincide, the symmetry is enlarged to $U(1) \times SU(2) = U(2)$ ($d=4$). Let us focus for a moment on the latter, we can take $\lambda_1 = \lambda_2$. The only holographic energy--momentum tensor consistent with the symmetry is
\begin{equation}
    T = \alpha\left(
    \lambda_1^{-1} \sigma_1^2 + \lambda_1^{-1} \sigma_2^2-2\lambda_3^{-1} \sigma_3^2
    \right).
\end{equation}
This is a two parameter family of initial data parameterized by $\frac{\lambda_1}{\lambda_3}$ and $\frac{\alpha}{\sqrt{\lambda_1}}$. It describes Taub-NUT-de Sitter and the two parameters correspond to the NUT parameter $l$ and mass parameter $m$, respectively. Notice, that $m$ is not a physical mass -- since $\mathcal{I}^+$ is topologically a sphere, any conserved charge
\begin{equation}
    Q_\xi[C] = \oint_C T_{ab}n^a \xi^b \sqrt{h} d^2x
\end{equation}
associated with a symmetry generator $\xi$ and a surface $C$ must vanish identically. Moreover, notice that the Killing horizon in the Killing horizon in this solution is in fact a Cauchy horizon and so it is a breakdown of the unique evolution. Notice that if we allow $\lambda_1 \neq \lambda_2$ the possible $T_{ab}$ are
\begin{equation}
    T_{ab} dx^a dx^b = \sum_{i,j=1}^3 h_{ij} \sigma_i \sigma_j
\end{equation}
subject to the usual constraints. They all correspond to Bianchi IX universes.
\subsection{$\mathbb{R}^3$}
Let us start with a metric $q^{(0)}$ which is not conformally flat. It follows than that (in an appropriate conformal frame) all symmetries acts as isometries. It is well-known that a non-maximally symmetric metric can be at most $4$-dimensional. Thus, one can generate plenty of examples even with $T_{ab}=0$. In particular, all homogeneous metrics (up to accidental additional symmetries) on $\mathbb{R}^3$ could be used and they are already classified \cite{Patragenaru}.
\subsubsection{Euclidean space} \noindent
Let us now discuss the case when the metric induced on $\mathcal{I}^+$ is flat. Then, our starting point is $7$-dimensional group\footnote{Notice that since we discuss the group, we must exlude the special conformal transformations since they move the point at infinity} of $\mathbb{R}^3$.  We want to break symmetry explicitly through the introduction of a non-trivial $T_{ab}$. All complete CKVFs on $\mathbb{R}^3$ are of the form:
\begin{equation}
    \xi = \left( p^i + r^i_{\ j} x^j + S x^i \right) \partial_i.
\end{equation}
As follows from our previous discussions, we must put $S=0$ -- otherwise $T_{ab}$ would have to vanish. We are left only with Killing vectors of the flat metric. Thus, we are now looking for $T_{ab}$ such that
\begin{equation}
    \mathcal{L}_\xi T_{ab} = 0,
\end{equation}
where $\xi$ is in a proper subalgebra of the Euclidean algebra. That subalgebra cannot be $6$ dimensional, since then $T_{ab}$ would be proportional to $\delta_{ab}$ and hence vanishing. Moreover, there are no $5$-dimensional subalgebras. Thus, the smallest possible one is $4$-dimensional and there is only one (up to an isomorphism). It is generated by $\partial_x, \partial_y, \partial_z, x\partial_y - y \partial_z$. The most general $T_{ab}$ it preserves is of the form
\begin{equation}
    T = a dx^2 + a dy^2 - 2a dz^2,
\end{equation}
where $a$ is an arbitrary constant. Such initial data correspond to the Bianchi I cosmology with an additional axial symmetry. \\
There are several $3$-dimensional subalgebras. We obviously have an algebra of translations (isomorphic to $\mathbb{R}^3$). Clearly it preserves $T_{ab}$ of the form
\begin{equation}
    T = h_{ij} dx^i dx^j,
\end{equation}
where $h_{xx}+h_{yy}+h_{zz} = 0$ and all $h_{ij}$ are constant. Such initial data corresponds to (now more general) Bianchi I universes. \\
Different $3$ dimensional algebras are $so(3)$, an euclidean algebra of a plane (spanned by $\partial_x, \partial_y$ and $x\partial_y - y \partial_x$) and a helical algebra ((spanned by $\partial_x, \partial_y$ and $\alpha \partial_z + x\partial_y - y \partial_x$). It is easy to see that the only $T_{ab}$ preserved by these symmetries (and no other) is simply zero. \\
It is also easy to construct examples with lower symmetry. In particular, an algebra $\mathbb{R}^2$ spanned by $\partial_x, \partial_y$ can be obtained by a choice
\begin{equation}
    T = h_{ij}(z) dx^i dx^j,
\end{equation}
where $h_{xx}+h_{yy}+h_{zz} = 0$ and $h_{iz}$ are constants. It is not clear to us whether such solutions occur in any physically interesting scenarios.
\subsection{Cylinder $\mathbb{R} \times S^2$} \noindent
Yet another possible topology of $\mathcal{I}^+$ is a cylinder $\mathbb{R} \times S^2$. This should describe BH spacetimes.
Let us take $q^{(0)}$ to be conformally flat and given by
\begin{equation}
    q^{(0)} = du^2 + \ell^2 d\theta^2 + \ell^2 \sin^2 \theta d\phi^2.
\end{equation}
Among all CKVFs, only $4$-dimensional subalgebra is complete in this case and is generated by $\partial_u$ and rotations. Thus, there are no solutions with more than $4$d isometries. The most general form of $T_{ab}$ consistent with those symmetries is
\begin{equation}
    T = a\left(2 du^2 - \ell^{-2} d\theta^2 - \ell^{-2} \sin^2 \theta d\phi^2\right),
\end{equation}
where $a$ is an arbitrary constant. It clearly corresponds to the Schwarzschild--de Sitter spacetime. The only three dimensional subalgebra is $so(3)$ and again Schwarschild--de Sitter is the only example (as stated by the Birkhoff theorem). The only two dimensional subalgebra is spanned by $\partial_u$ and $\partial_\phi$. If we consider a one-dimensional algebra spanned by $\partial_u$, we can write down the most general form of $T_{ab}$:
\begin{equation}
    T = \ell^{-2} (\mathring{\Delta}+2) \Phi du^2 + 2 \epsilon_{AB}\mathring{D}^B \chi dx^A + (\mathring{D}_A \mathring{D}_B - (\mathring{\Delta} + 1) \mathring{\gamma}_{AB}) \Phi dx^A dx^B,
\end{equation}
where $\Phi$ and $\chi$ are arbitrary ($u$-independent) functions on $S^2$. If they are additionally axially symmetric we are back to the previous case. Notice that Kerr-de Sitter can be put in this form \cite{Kolanowski:2021hwo}.
\section{Discussion}  In this paper we considered asymptotically de Sitter spacetimes that satisfied the vacuum Einstein equations in a neighborhood of the null infinity $\mathcal{I}^+$. We assumed an existence of KVFs in that neighborhood and studied their properties. The key elements of our analysis were the initial Cauchy problem at $\mathcal{I}^+$ and the Fefferman--Graham expansion.\\
The first result is Theorem \ref{Theorem}. It states that if there is a Killing horizon that intersects ${\mathcal I}^+$ (after a conformal completion), then in a neighborhood of the intersection point, the spacetime is isometric to the de Sitter spacetime. We explain now, why the theorem is true.  To begin with, every KVF $X$ admitted by a neighborhood of  ${\mathcal I}^+$ turns out to be tangent to ${\mathcal I}^+$.  Conversely, a vector field $\xi$ tangent to  ${\mathcal I}^+$ is a restriction of a KVF defined in a neighborhood of ${\mathcal I}^+$ if and only it is a symmetry of the data  that determines a solution of Einstein's equations in the neighborhood. The data is a pair:  the induced metric tensor $q^{(0)}$ and holographic stress-energy tensor $T$ defined up to the conformal transformations (\ref{data}). Next, it turns out that the intersection of the Killing horizon with ${\mathcal I}^+$ is an isolated point, a zero of the vector field $\xi$. That follows from the properties of essential CKVF of $3$ dimensional conformal geometries. On the other hand, if $\xi$  were a non-essential CKVF, then the solution $X$ of the Killing equation  that is determined in the neighborhood would be spacelike (hence, without a Killing horizon). The last step of the reasoning is the construction of the scalar $q^{(0) ac} q^{(0) bd} T_{ab} T_{cd} \left( q^{(0)}_{ij} \xi^i \xi^j \right)^3$ that is shown to be zero on a neighborhood of the zero of the symmetry $\xi$. 
\\
We have systematically investigated possible symmetries of the initial data. In particular, we have proven that they exhibit the gap phenomenon with the submaximal symmetry being only $4$-dimensional. We have also shown that this case reduces to the homogeneous geometry on $\mathcal{I}^+$ and as such is much easier to understand. In particular, if the null infinity is topologically a sphere, then the solution is necessarily Taub-NUT-de Sitter. We have also provided a lot of examples in other situations. Hopefully they can be useful as a starting point for the perturbative treatment of the gravitational radiation.

\section*{Acknowledgments}\noindent This work was supported by OPUS 2017/27/B/ST2/02806 of the Polish National Science Center.

\bibliography{bibl.bib}
\end{document}